\documentclass[journal]{IEEEtran}

\usepackage{cite}

\usepackage{amsmath}
\usepackage{enumitem}
\usepackage{amsthm}
\usepackage{amssymb}
\usepackage{graphicx}
\usepackage{latexsym}
\usepackage{subfigure}
\usepackage{color}
\usepackage{epsfig}
\usepackage{dsfont}
\usepackage{fancyhdr}
\usepackage{pst-all}
\usepackage[scanall]{psfrag}
\usepackage{multirow}
\usepackage{booktabs}
\usepackage{textcomp}
\usepackage{siunitx}
\usepackage{textcomp}
\usepackage{courier}
\usepackage{cuted}

\usepackage{amsfonts,amssymb}
\usepackage[german,english]{babel}
\usepackage{url}
\usepackage[colorlinks=true,
            linkcolor=black,     
            citecolor=black,    
             urlcolor=blue,
            ]{hyperref}

\renewenvironment{proof}{\vspace{.1cm}\noindent{\sc Proof.}\hspace{0.10cm}\,\,}{\hfill$\blacksquare$ } 
\newtheorem{theorem}            {Theorem}[section] 
\newtheorem{definition}         [theorem]{Definition}

\newtheorem{lemma}              [theorem]{Lemma}

\newcommand{\nline}{{\mathbb N}}

\newcommand{\sbm}[1]{\left[\begin{smallmatrix} #1
	\end{smallmatrix}\right]}

\newcommand{\rfb}[1]{\mbox{\rm
		(\eref{#1})}\ifx\undefined\stillediting\else:\fbox{$#1$}\fi}

\newcommand{\bluff}{{\hbox{\raise 15pt \hbox{\hskip 0.5pt}}}}

\newfont{\roma}{cmr10 scaled 1200}

\usepackage{bm}

\usepackage{algorithm}
\usepackage{algorithmicx}
\usepackage{algpseudocode}
\usepackage{amsmath}
\usepackage{amssymb}
 
\floatname{algorithm}{Algorithm}

\ifCLASSINFOpdf
\else
\fi
\begin{document}
\title{Flexible Distributed Flocking Control for Multi-agent Unicycle Systems}
        
\author{Tinghua Li and Bayu Jayawardhana 
	\thanks{Tinghua Li and Bayu Jayawardhana are with DTPA, ENTEG, Faculty  of  Science  and  Engineering,  University  of  Groningen, The Netherlands
		\{\tt\small t.li, b.jayawardhana\}@rug.nl.}%
}

\maketitle

\vspace{12pt}
\begin{abstract}
Currently, the general aim of flocking and formation control laws for multi-agent systems is to form and maintain a rigid configuration, such as, the $\alpha$-lattices in flocking control methods, where the desired distance between each pair of connected agents is fixed. This introduces a scalability issue for large-scale deployment of agents due to unrealizable geometrical constraints and the constant need of centralized orchestrator to ensure the formation graph rigidity. This paper presents a flexible distributed flocking cohesion algorithm for nonholonomic multi-agent systems. The desired geometry configuration between each pair of agents is adaptive and flexible. The distributed flocking goal is achieved using limited information exchange (i.e., the local field gradient) between connected neighbor agents and it does not rely on any other motion variables measurements, such as (relative) position, velocity, or acceleration. Additionally, the flexible flocking scheme with safety is considered so that the agents with limited sensing capability are able to
maintain the connectedness of communication topology at all time and avoid inter-agent collisions. The stability analysis of the proposed methods is presented along with numerical simulation results to show their effectiveness. 
\end{abstract}


\section{Introduction}

The distributed cooperation of multiple autonomous agents has received a great amount of attention in various research fields and disciplines such as robotics\cite{Moshtagh-i,Saulnier}, autonomous vehicles \cite{besselink,turri}, unmanned aerial vehicles \cite{sun2021collaborative}, \cite{de2019flexible}, physics\cite{Frasca}, and computer science\cite{Reynolds}. Based on the interaction between inter-agent in the network, multi-agent systems must achieve certain tasks or missions (e.g., formation, obstacle avoidance and area coverage) via distributed control. Inspired by the bacterial swarming \cite{Kearns} in planktonic environment and in underwater fish schooling \cite{Okubo}, flocking is one of the typical collective behavior motions which involves coordination and cooperation among agents, and has been widely studied in the engineering field \cite{Zavlanos,Loizou}, where the Reynolds's rules \cite{Reynolds} are commonly used.

To achieve the aggregation and cohesion of the flocking group, the multi-agent systems must establish and maintain a connected network,  
where each pair of connected agents can share/obtain (relative) information unilaterally or bidirectionally, and they can regulate their interaction.  
In most cases, the desired flocking configuration is defined as a \emph{rigid} structure where each agent is equally distanced from all of its neighbors (e.g., the \emph{$\alpha$-lattices} in \cite{olfati2006flocking}). In practice, such equal inter-agent spacing requirement for flocking can be restrictive and may lead to a well-posedness issue. Firstly, it restricts its application in an environment where other constraints must be taken into account (e.g., 
obstacles, connectivity, and dynamic environment). Secondly, geometrical constraint prevents a large group of agents to achieve equal inter-agent spacing requirement with arbitrary network topology; hence it may lead to a well-posedness and scalability issue. 
An alternative to this framework is the \emph{quasi-$\alpha$ lattices} approach where the flocking is regulated through a potential energy function and the deviation from its minimum corresponds to 
the degree to which the final configuration differs from an \emph{$\alpha$-lattice}. However, this approach still assumes an identical inter-agent distance for all edges. In contrast, graph rigidity theory has been used for defining arbitrary formation shape with different inter-agent distance and consequently for designing  distributed formation controllers. This framework is not suitable for achieving distributed flocking control as the set of inter-agent distances must be known and set {\it a priori} 
\cite{Deghat}.

Motivated by the constant time-gap spacing policy in vehicle platooning which is velocity-dependent with each agent's preceding one \cite{ploeg2011,ploeg2013}, we propose a leaderless distributed adaptive spacing policy for the undirected distributed flocking cohesion control, 
which ensures that the final formation shape is well-posed and attractive. Due to the adaptive spacing policy, each inter-agent spacing gap is no longer required to be constant and equal. 
The adaptive spacing policy relies solely on each agent's local information and limited communication with neighbors. In particular, it is not velocity-dependent so that the neighbors' velocity is not required.

Another practical design constraint studied in this paper is the connectivity preservation of communication topology. In practical applications, communication topology can be broken, which can be due to agents' limited sensing capabilities and 
may result in flocking splitting and fragmentation during motion evolution \cite{olfati2006flocking}. 
Correspondingly, several approaches have been presented in literature to avoid undesirable group disconnection and communication loss. Some  local and global connectivity preservation approaches are the potential function-based, algebraic graph theory-based, and more recently optimization-based approaches. For the latter approach, 
the use of control barrier functions (CBFs) has been proposed 
\cite{Ames_CDC,Wang_CBF,Capelli_2020,Capelli_2021,Li_Flocking}. 
However, it is known that the quadratic programming (QP) associated to CBFs method can pose some feasibility challenges due to the tight control bounds and a dense quantity of agents \cite{Ames_TAC,Xiao_auto}. In this case, the application of potential function-based approach is preferred in practice to 
preserve 
the connected edge. 
The edge connectivity can be preserved by using unbounded potential functions, which can lead to large control effort \cite{Zavlanos}-\cite{Jia}, or by using  bounded ones \cite{Wen,Li_bounded,Ajorlou} whose gradient is in the direction of preserving the edge connectivity. 
A global connectivity maintenance method is studied in \cite{Kim,Poonawala} to maximize the second smallest eigenvalue of the graph Laplacian, where each agent requires 
the global information of the network. 
In this paper, we focus on the local edge connectivity for our distributed flocking cohesion control that is easy to implement. 



In summary, 
the main contributions of this paper are as follows. 
Distributed flocking cohesion control with adaptive spacing policy (ASP) is proposed for 
    multi-agent unicycle system which ensures practical flocking cohesion of the group to a desired configuration with well-posed geometry. It uses only local information and limited inter-agent communication, and it does not require neighbors' motion variables (e.g. velocity, acceleration, and orientation), relative distance, or the global position of each static obstacle. 
 The proposed ASP can be deployed in both the static and switching graph. Moreover we show that the constraints of inter-agent collision avoidance and local connectivity preservation can directly be incorporated. 

The structure of the paper is as follows. Section \ref{sec:problem} provides the mathematical background and problem formulation. In Section  \ref{sec:control}, we introduce the design of adaptive flocking policy and cohesion control law, and provide rigorous convergence proof in the static graph case. In Section \ref{sec:dynamic}, we extend 
the algorithm to the dynamic graph case and discuss the constraints of inter-agent collision and connectivity preservation. The simulations are shown in Section \ref{sec:simulation} and the conclusions are presented in Section \ref{sec:conclusion}.


\section{preliminaries}\label{sec:problem}
Consider a group of $N$ unicycle agents initialized in an unknown scalar field $J(x,y)$ in $\mathbb{R}^2$, which is a twice continuously differentiable function. 
The kinematic dynamics of each unicycle-model agent $i\in \left \{ 1,2,3...,N \right \}$ is given by
\begin{equation}\label{eq:unicycle_model}
\begin{bmatrix}
\dot{x}_i \\
\dot{y}_i  \\
\dot{\theta}_i
\end{bmatrix}=\begin{bmatrix}
v_i\cos(\theta_i) \\
v_i\sin(\theta_i)  \\
\omega_i
\end{bmatrix},
\end{equation}
where $v_i$ and $\omega_i$ are the velocity and angular velocity inputs, respectively. It is assumed that 
each agent $i$ is equipped with on-board sensor systems that are able to measure the field gradient $\nabla J(x_i,y_i)$. For simplifying the notation, we denote $X_i:=\nabla J(x_i,y_i)$ and $X = \text{col}(X_i)$, whenever it is clear from the context. 
We further assume that every agent has a local reference frame whose origin is located in its center of mass and it shares the same North-East (NE) orientation with other agents. Thus 
the locally measured field gradient vector can be converted into this common frame $^{c}\textrm{}\sum $. 

\subsection{Communication Topology}\label{sec:int_topo}
The undirected communication graph $\mathcal{G}=(\mathcal{V},\mathcal{E})$ consists of the vertex set $\mathcal{V}=\left \{ 1,2,...,N \right \}$ representing agents and the corresponding edge set $\mathcal{E} \subseteq  \left \{  (i,j): i\in \mathcal{V}, j \in \mathcal{N}_i \right \}$ where $\mathcal{N}_i$ is the neighbor set of agent $i$. The graph $\mathcal G$ is called {\em connected} if for any pair of agents $i$ and $j$, there is a path that connects them. The edge $(i,j)\in\mathcal{E}$ indicates that agents $i$ and $j$ are connected by bidirectional information channel and they can exchange their field gradient measurement 
$^cX_i$ and $^cX_j$ in the common NE frame $^{c}\textrm{}\sum $ with each other (i.e., $(i,j) \in \mathcal{E}\Leftrightarrow  (j,i) \in \mathcal{E}$). For simplicity, we denote the field gradient in the frame $^{c}\textrm{}\sum $ by $X_i$.  
Based on these mild assumptions, our proposed design of flocking cohesion controller with ASP does not need other information exchanges (e.g., orientation, velocity, acceleration, etc.) which are more difficult to implement distributedly using local frame and sensors. 
The corresponding adjacency matrix $A= \left [ a_{ij}\right ]$ of $\mathcal G$ is defined by 
$a_{ij} =1 \Leftrightarrow (i,j)\in \mathcal{E} $, or otherwise $a_{ij} = 0$.

\subsection{Flexible Flocking Cohesion}
In the widely accepted Reynolds's flocking rule, the \emph{Flocking Centering} is defined as each agent attempting to stay close to nearby flockmates\cite{Reynolds}, and it is generalized by Olfati-Saber in \cite{olfati2006flocking} using 
a lattice-type structure to model a desired geometry, where each agent is equally distanced from all of its neighbors \cite{olfati2006flocking}. Among these works, the desired inter-agent gap $D^*$ is generally defined as identical and constant during the evolution. 

In this paper, the agents aim to realize a flexible flocking cohesion such that each edge length (i.e., inter-agent gap) obeys an adaptive spacing policy 
\begin{equation}
     \mu_{ij}(t) \rightarrow D^*_{ij}(t)\in\mathbb{R}_+, \quad \forall (i,j)\in\mathcal{E}
\end{equation}
where $\mu_{ij} $ is the gradient difference of the neighboring agents $i$ and $j$ with $\left \| \cdot  \right \|$ denoting the Euclidean norm and is defined by
\begin{equation}\label{eq:muij}
    \mu_{ij}(X_i, X_j) =  \left \|   X_{j}- X_{i} \right \|,
\end{equation}
and $D^*_{ij}$ is a target time-varying inter-agent spacing gap, which will be designed later to ensure the fulfillment of geometrical constraints during flocking. 
In this paper, we consider the \emph{spacing gap} 
in a generic formulation using 
the field gradient difference 
$\mu_{ij}(X_i, X_j)$ 
as in \eqref{eq:muij}. This formulation includes the standard flocking problem where $J_i$ is a quadratic function, so that $\mu_{ij}$ is given by the Euclidean distance between agent $i$ and $j$. It can also represent a multi-agent systems coverage control problem where all agents must cover a certain field of physical variables (e.g., chemical concentrations). 

\subsection{Problem Formulation}
Given the communication graph $\mathcal{G}$ in a scalar field $J(x,y)$, 
we study the design of a distributed flexible flocking control law with an adaptive spacing policy, where each agent in the network relies only on its own and neighbors' local field gradient measurements. \\[0.1cm] 
\textbf{Flexible Flocking Cohesion Problem with Adaptive Inter-agent Spacing Policy and Static Graph:} 
     \begin{enumerate}
         \item {\bf Adaptive Inter-agent Spacing Policy}: Design a state-dependent ASP policy $D^*_{ij}$ is 
         for each connected agent pair in the 2D configuration space $(x,y)\in\mathbb{R}^2$ such that the inter-agent spacing 
         is allowed to be dynamic and inconsistent to ensure well-posedness of the geometrical constraints of the flock: 
         \begin{equation}\label{eq:ASP}
             D^*_{ij} = F(X_i, X_j), \quad \forall(i,j)\in\mathcal{E}
         \end{equation}
         
         \item {\bf ASP-based Flexible Flocking Cohesion}: 
         Design a distributed flocking cohesion control law $u_{i}$, $i\in\mathcal{V}$ such that 
         the multi-agent unicycle system \eqref{eq:unicycle_model} achieves the flexible flocking cohesion where each agent converges to a desired configuration with well-posed geometry,
        i.e.,
            \begin{equation}\label{eq:convergence-1}
               \lim_{t\to\infty} \sum_{(i,j)\in\mathcal{E}} e_{ij}(t) =  \lim_{t\to\infty} \sum_{(i,j)\in\mathcal{E}}
                \Big(\mu_{ij}(t), D^*_{ij}(t)\Big) \rightarrow  0
            \end{equation} 
            where $D^*_{ij}(t)$ is the adaptive inter-agent spacing policy for edge $(i,j)\in\mathcal{E}$ as in \eqref{eq:ASP}.
 \end{enumerate}

As an extension of the above formulated problem, we will also investigate 
the following two extended sub-problems involving additional hard constraints and with dynamic interaction graph. 
Consider the aforementioned setup of flexible distributed flocking cohesion control law $u_i$ with adaptive inter-agent spacing policy $D_{ij}^*$ and with dynamic communication topology $\mathcal{G}(t)=\left\{\mathcal{V},\mathcal{E}(t)\right\}$, where each agent has limited sensing and interaction range $r$. In this case, the constrained flocking cohesion control problem is associated to the following additional sub-problems.  
\begin{itemize}
     \item \textbf{Local Connectivity Preservation:} 
    For each agent $i\in\mathcal V$, the following implication holds: if the $j$-th agent becomes a neighbor of agent $i$ at time $t_k > t_0$ (i.e. $\mu_{ij}(t_k)<r$) then the pair $(i,j)$ remains connected for all $t\geq t_k$ (i.e., $\mu_{ij}(t)<r$ for all $t\geq t_k$). 
     
     \item \textbf{Inter-agent Collision Avoidance:} All agents stay within the safe region among each other, i.e.,
\begin{equation}\label{task:collision}
          \mu_{ij}(t)  >0,  \quad \forall  (i,j)\in\mathcal{E}(t), \, t\geq t_0.
     \end{equation}    
\end{itemize}


\section{Spacing and Control Design for Static Graph}\label{sec:control}
In this section, we first present the design requirements of the adaptive spacing policy (ASP), and then propose the ASP-based flocking cohesion control method in the static and connected graph $\mathcal{G}$.  

\subsection{Adaptive Inter-agent Spacing Policy}\label{sec:ASP_static}
Given the group of $N$ unicycle agents as in 
\eqref{eq:unicycle_model}, each pair of agents $(i,j)\in\mathcal{E}$ assumes a nominal spacing value $d_{i,\text{nom}}, d_{j,\text{nom}}\in\mathbb{R}_{+}$ between them, which can be identical $d_{i,\text{nom}}= d_{j,\text{nom}}$ or inconsistent with each other $d_{i,\text{nom}}\neq  d_{j,\text{nom}}$. Even when they are identical, we do not assume that they admit a well-posed geometrical shape. 
For simplifying the presentation, we assume throughout the paper an identical and consistent $d_{\text{nom}}$ for all edges. 
Accordingly we define the adaptively-desired inter-agent gap $D^*_{ij}$ with a scaling factor $s_{ij}$ by
\begin{equation}\label{eq:D*}
    D^*_{ij}(t) = d_{\text{nom}}s_{ij}(t) , \quad (i,j)\in\mathcal{E},
\end{equation}
and the corresponding spacing error $e_{ij}$ 
is given by
 \begin{equation}\label{eq:eij}
    e_{ij}= \mu_{ij}-D^*_{ij},
\end{equation}
where $\mu_{ij}$ is the field gradient difference as in \eqref{eq:muij}. For designing a state-dependent policy $D^*_{ij}(t)$, we will impose the following design conditions: 
the scaling factor $s_{ij}$ is a positive and bounded state-dependent function for the pairing agents in the edge $(i,j)\in\mathcal{E}$, and it is constant whenever  
$e_{ij}=0$. The latter condition 
implies that a well-posed geometry is formed and thus the scaling factor is fixed. 
The following 
scaling factor, which depends on an auxiliary state variable 
$d_{ij}$, fulfills the aforementioned design conditions:  
 \begin{equation}\label{eq:sij}
        \begin{aligned}
             s_{ij} &=  e^{\lambda \tanh{\left(\frac{d_{ij}-d_{\text{nom}}}{2}\right)}} \in \left (e^{-\lambda} ,e^{\lambda}\right ) 
        \end{aligned}
        \end{equation}
      
    \begin{equation}\label{eq:d_dot}
       \dot  d_{ij} =  \tanh{\left(\frac{e_{ij}}{2}\right)} 
       \in \left ( -1,1 \right ), \quad  d_{ij}(t_0) = d_{\text{nom}}
    \end{equation}
where $\lambda\in\mathbb{R_+}$ is a positive scaling parameter. For a static graph as considered in this section, $\lambda$ can be any positive constant. However, for a dynamic graph as treated in Section \ref{sec:dynamic}, $\lambda$ needs to be designed properly to maintain connectivity preservation and it will be presented in Section \ref{sec:dynamic}. 
We note both $s_{ij}$ and $\dot d_{ij}$ are set bounded in order to avoid over-expansion or constriction. Since $s_{ij}>e^{-\lambda}>0$, it  follows from \eqref{eq:D*} that $D^*_{ij}>d_{\text{nom}}e^{-\lambda}>0$, which implies that the adaptive spacing rule of $s_{ij}$ and $d_{ij}$ in \eqref{eq:sij}--\eqref{eq:d_dot} always leads to an admissible desired gradient difference $D^*_{ij}$. As will be shown later in our main technical result, we will use $s_{ij}$ as the state variable of the adaptive spacing law and its time-derivative can be computed directly using the smoothness property of the tangential hyperbolic function. 



\subsection{Flexible Flocking Cohesion with Adaptive Spacing Policy}\label{sec:Flocking_Spacing}
The proposed distributed flexible flocking controller $u_i=\sbm{v_i\\\omega_i}$ for each agent $i\in\mathcal V$ is given by 
\begin{equation}\label{eq:flocking_controller}   
        \begin{aligned}
                 v_i &= -\frac{K_f}{N_i} \vec{o}_i \nabla^2 J_{i} \left(\frac{\partial P_i}{\partial X_i} \right)^\top 
                \\ \omega_{i} &= \frac{K_f}{N_i}  \vec{o}^\perp_i  \nabla^2 J_{i} \left(\frac{\partial P_i}{\partial X_i}\right)^\top
        \end{aligned}
        \end{equation}
where $N_i =\left | \mathcal{N}_i \right | $ is the quantity of the neighbors around agent $i$, $\vec{o}_i = \sbm{ \cos(\theta_{i}) & \sin(\theta_{i})}$ and $\vec{o}^\perp_i = \sbm{ \sin(\theta_{i}) & -\cos(\theta_{i}) } $ are the unit (orthogonal) orientation vectors. 

In addition, $P_i$ denotes the agent's flocking potential, which is assumed to be a strictly convex and differentiable function of $e_{ij}$, $j\in \mathcal N_i$ with a minimum at $0$. 
One simple example of the potential function is
\begin{equation}\label{eq:poten}
    P_i =\frac{1}{2} \sum_{k\in\mathcal{V}\setminus \left\{i\right\}}a_{ik}e^2_{ik} =\frac{1}{2}\sum_{j\in\mathcal{N}_i}e^2_{ij},
\end{equation}
where $a_{ik}$ is the element of adjacency matrix $A= \left [ a_{ik}\right ]$ of the graph. 

We emphasize again that each pair of connected agents are able to communicate with each other and exchange their local measurements of the field gradient. Therefore the control law \eqref{eq:flocking_controller} can be implemented in each local frame of agent $i$, and relies only on the neighbors' field gradient measurement. There is no need to receive the neighbor's motion variables (e.g. acceleration, angular velocity, orientation, etc.). 

\begin{theorem}\label{thm:static}
      Consider a multi-agent unicycle system  \eqref{eq:unicycle_model} with 
      the distributed flocking controller \eqref{eq:flocking_controller}. Suppose that the undirected communication graph $\mathcal{G}$ is static and connected during evolution. Then given the state-dependent adaptive spacing policy \eqref{eq:D*}, the agents are able to form a flexible configuration and converge to a well-posed geometry with desired inter-agent space, i.e.
        \begin{equation}
            \lim_{t\to\infty} \sum_{(i,j)\in\mathcal{E}}
         e_{ij} = \lim_{t\to\infty} \sum_{(i,j)\in\mathcal{E}} \Big(\mu_{ij}-D^*_{ij} \Big)\rightarrow  0
        \end{equation}
      holds.
\end{theorem}   

\begin{proof}
    In order to apply the well-known LaSalle's invariance principle \cite{sastry2013nonlinear} to the closed-loop multi-agent unicycle systems, 
    we first need to establish that all state trajectories are bounded for the corresponding autonomous system. As the state $\theta_i$ in unicycle model \eqref{eq:unicycle_model} is evidently unbounded, we redefine each agent's dynamics with extended states by 
    \begin{equation}\label{eq:state_z}
        z_i = \begin{bmatrix} x_i & y_i & \cos(\theta_i) & \sin(\theta_i)\end{bmatrix} 
    \end{equation}
    which follows that     \begin{equation}\label{eq:state_z_dot}
     \begin{aligned}
    \dot z_i =\sbm{
     v_i\cos(\theta_i) \\
     v_i\sin(\theta_i)  \\
    -\cos(\theta_i)\sin(\theta_i)\dot \theta_i  \\
    \cos(\theta_i)\sin(\theta_i)\dot \theta_i} & = \sbm{\cos(\theta_i) & 0 \\
    \sin(\theta_i)&0 \\
    0&-\cos(\theta_i)\sin(\theta_i)  \\
     0&\cos(\theta_i)\sin(\theta_i) } \sbm{ v_i \\ \omega_i } .
     \end{aligned}
   \end{equation}
Using the flocking controller \eqref{eq:flocking_controller}, all control inputs are state-dependent which makes the closed-loop systems autonomous. 
Let us consider 
the following Lyapunov function for the multi-agent system 
with the extended state \eqref{eq:state_z} 
\begin{equation}\label{eq:Lyapunov}
    V(z)   = \sum_{i\in \mathcal{V} }P_i +\frac{1}{2}\sum_{i\in\mathcal{V}}z^2_{i3} +\frac{1}{2}\sum_{i\in\mathcal{V}}z^2_{i4}
\end{equation}
where $P_i$ is a strictly convex and differentiable potential function of $e_{ij}$, $j\in\mathcal N_i$. Note that the first term of \eqref{eq:Lyapunov} can be expressed as a function of $z_{i1}$ and $z_{i2}$ for all $i\in\mathcal V$, so that $V$ is a function of all state variables $z$.  
As the communication graph $\mathcal{G}$ is assumed to be  static and connected for all time, the time-derivative of $V(z)$ is given by 
\begin{equation}\label{eq:v_dot_expres}
\begin{aligned}
    \dot V &= \sum_{i\in\mathcal{V}} \dot P_i + \underbrace{\sum_{i\in\mathcal{V}} \left( -\cos(\theta_i)\sin(\theta_i)\dot \theta  + \cos(\theta_i)\sin(\theta_i)\dot \theta \right)}_{=0}
    \\& =\sum_{(i,j)\in\mathcal{E}} \frac{\partial  P_i}{\partial  e_{ij}}\dot e_{ij}
\end{aligned}
\end{equation}
where $\dot e_{ij}$ can be rewritten by using the ASP \eqref{eq:sij} and \eqref{eq:d_dot} as follows
\begin{equation}\label{eq:e_dot}
\begin{aligned}
    \dot e_{ij} &= \frac{\partial e_{ij}}{\partial \mu_{ij}} \dot \mu_{ij} + \frac{\partial e_{ij}}{\partial s_{ij}} \dot s_{ij}.
\end{aligned}
\end{equation} 
On the one hand, from the definition of $e_{ij}$, it follows that  $\frac{\partial e_{ij}}{\partial \mu_{ij}}=1$ 
and $\frac{\partial e_{ij}}{\partial s_{ij}}  = -d_{\text{nom}}$. On the other hand, the time-derivative of the scaling multiplier $s_{ij}$ is given by 
\begin{equation}\label{eq:sij_dot}
\begin{aligned}
      \dot s_{ij} &=
      e^{\lambda \tanh{\left (\frac{d_{ij} -d_{\text{nom}}}{2}\right )}}\frac{2\lambda e^{-(d_{ij}-d_{\text{nom}})}}{(1+e^{-(d_{ij}-d_{\text{nom}})})^2}\dot d_{ij}
      \\& = e^{\lambda \tanh{\left (\frac{d_{ij} -d_{\text{nom}}}{2}\right )}} \frac{2\lambda e^{-(d_{ij}-d_{\text{nom}})}}{(1+e^{-(d_{ij}-d_{\text{nom}})})^2} \tanh{\left(\frac{e_{ij}}{2}\right)}.
 \end{aligned}
\end{equation}
With regards to $\dot \mu_{ij}$, it can be computed by
\begin{equation}
       \dot  \mu_{ij} = \frac{\mathrm{d}  \mu_{ij}}{\mathrm{d}(X_j - X_i)} (\dot X_j - \dot X_i) 
\end{equation}
where $\frac{\mathrm{d}  \mu_{ij}}{\mathrm{d}(X_j - X_i)}=\frac{X_j - X_i}{\left\|X_j-X_i\right\|}$ is denoted by a  unit vector $\Vec{o}_{\mu_{ij}}=\sbm{\cos(\beta_{ij}) & \sin(\beta_{ij})}$,  
and $\beta_{ij}$ is defined by 
\begin{equation}\label{eq:beta}
    \beta_{ij} = \text{tan}^{-1}\frac{X_{j,y} - X_{i,y} }{ X_{j,x} - X_{i,x}}.
\end{equation}
Following a similar computation to the term of $ \frac{\partial P_i}{\partial X_i}$ in 
\eqref{eq:flocking_controller}, we have that 
$ \frac{\partial P_i}{\partial X_i} = -\sum_{(i,j)\in\mathcal{E}}\frac{\partial P_i}{\partial e_{ij}}\Vec{o}_{\mu_{ij}}$. 
Accordingly, it follows from 
\eqref{eq:v_dot_expres} that 
\begin{equation}\label{eq:V_dot}
\begin{aligned}
     \dot V &= \sum_{(i,j)\in\mathcal{E}} \frac{\partial  P_i}{\partial  e_{ij}} \left(\frac{\partial e_{ij}}{\partial \mu_{ij}} \dot \mu_{ij} + \frac{\partial e_{ij}}{\partial s_{ij}} \dot s_{ij} \right)
     \\& = \sum_{(i,j)\in\mathcal{E}} \frac{\partial  P_i}{\partial  e_{ij}} \left(\dot \mu_{ij} -d_{\text{nom}} \dot s_{ij} \right)
    \\& = \sum_{(i,j)\in\mathcal{E}} \frac{\partial  P_i}{\partial  e_{ij}} \Vec{o}_{\mu_{ij}}\left(\dot X_j - \dot X_i\right) - d_{\text{nom}}  \frac{\partial  P_i}{\partial  e_{ij}}\dot s_{ij}.   
    \end{aligned}
\end{equation}
Since each agent shares the same orientation of North-East reference frame, 
the connectivity of the undirected interconnection topology $\mathcal{G}$ (i.e., $j \in \mathcal{N}_i \Leftrightarrow  i \in \mathcal{N}_j$) implies that 
\begin{equation}\label{symmetry}
\left.\begin{matrix}
\begin{aligned}
     e_{ij}&= e_{ji} 
    \\ \Vec{o}_{\mu_{ij}} &= -\Vec{o}_{\mu_{ji}}
\end{aligned}
\end{matrix}\right\}.
\end{equation}
By hypothesis, the potential function $P_i$ 
is a strictly convex and differentiable function of the spacing error $e_{ij}$. Hence $\frac{\partial P_i}{\partial e_{ij}}$ is monotonically non-decreasing. 
By the symmetry $e_{ij} = e_{ji}$, 
$ \frac{\partial  P_i}{\partial e_{ij}}  = \frac{\partial  P_j}{\partial e_{ji}} $ holds for all edges $(i,j)\in\mathcal{E}$. Therefore, the first term in \eqref{eq:V_dot} can be expressed as 
\begin{equation}\label{eq:term_1_1}
\begin{aligned}
    \sum_{(i,j)\in\mathcal{E}} \frac{\partial  P_i}{\partial  e_{ij}} \Vec{o}_{\mu_{ij}}\left(\dot X_j - \dot X_i\right) &=  -2\sum_{(i,j)\in\mathcal{E}} \frac{\partial  P_i}{\partial  e_{ij}} \Vec{o}_{\mu_{ij}} \dot X_i
    \\& = 2 \sum_{(i,j)\in\mathcal{E}} \frac{\partial P_i}{\partial X_i} \dot X_i. 
\end{aligned}
\end{equation}
By the definition of $X_i$, its time derivative $\dot X_i $ can be expressed as a function of 
the controller $u_i = \sbm{v_i \\ \omega_i}$ as follows 
\begin{equation}
    \begin{aligned}
          \dot X_i & =  \nabla^2 J_{i} \sbm{ \dot x_i \\ \dot y_i}
                \\&=\nabla^2 J_{i} \sbm{\cos(\theta_i) \\ \sin(\theta_i)}v_i
                 \\& =  -\frac{K_f}{N_i} \nabla^2 J_{i}\vec{o}^{\top}_i \vec{o}_i \nabla^2 J_{i} \left(\frac{\partial P_i}{\partial X_i} \right)^\top .
    \end{aligned}
\end{equation}
Substituting this to 
\eqref{eq:term_1_1} gives us 
\begin{equation}\label{eq:term_1}
     \sum_{(i,j)\in\mathcal{E}} \frac{\partial  P_i}{\partial  \mu_{ij}} \Vec{o}_{\mu_{ij}}\left(\dot X_j - \dot X_i\right)  = -\sum_{(i,j)\in\mathcal{E}}\frac{2K_f}{N_i}\left( \frac{\partial P_i}{\partial X_i}\nabla^2 J_{i} \vec{o}^{\top}_i \right)^2. 
\end{equation}
For the second term in \eqref{eq:V_dot}, 
using \eqref{eq:sij_dot} we obtain 
\begin{equation}
    \sum_{(i,j)\in\mathcal{E}}\frac{\partial P_{i}}{\partial  e_{ij}}\dot s_{ij} = \sum_{(i,j)\in\mathcal{E}}\frac{\partial P_{i}}{\partial  e_{ij}} \frac{\mathrm{d} s_{ij}}{ \mathrm{d} d_{ij}} \dot d_{ij}
    =\sum_{(i,j)\in\mathcal{E}} H_{ij}K_{ij},  
\end{equation}
where
\begin{equation}
    \begin{aligned}
     H_{ij} &= \frac{\partial P_{i}}{\partial  e_{ij}}\dot d_{ij}=\frac{\partial P_{i}}{\partial e_{ij}} \tanh{\left(\frac{e_{ij}}{2}\right)}
       \\ K_{ij} & = \frac{\mathrm{d}  s_{ij}}{\mathrm{d} d_{ij}}= e^{\lambda \tanh{\left(\frac{d_{\text{nom}}-d_{ij}}{2}\right)}}\frac{2\lambda e^{-(d_{ij}-d_{\text{nom}})}}{(1+e^{-(d_{ij}-d_{\text{nom}})})^2} >0.
    \end{aligned}
\end{equation}
It is clear that $K_{ij}>0$ and 
the function $H_{ij}(e_{ij}) =\frac{\partial P_{i}}{\partial  e_{ij}}\dot d_{ij}$ with $\dot d_{ij}$ as in \eqref{eq:d_dot} is positive definite with $H_{ij}(0)=0$. 
Hence 
\begin{equation}\label{eq:term_2}
  \sum_{(i,j)\in\mathcal{E}} H_{ij}K_{ij}  \geq 0
\end{equation}
holds for all $z$. 
By combining \eqref{eq:term_1} and \eqref{eq:term_2} into \eqref{eq:V_dot}, it follows that 
\begin{align}
\nonumber
    \dot V(z) &=-\sum_{(i,j)\in\mathcal{E}}\frac{2K_f}{N_i}\left(\frac{\partial P_i}{\partial X_i}\nabla^2 J_{i} \vec{o}^{\top}_i \right)^2 
       -d_{\text{nom}}\sum_{(i,j)\in\mathcal{E}} H_{ij}K_{ij}  
        \\
        \label{eq:V_dot_final}
        & \leq 0.
\end{align}
It follows from this inequality that $V(z)$ is non-increasing  
and consequently, the state $z_i(t) = \sbm{x_i(t) & y_i(t)  & \cos(\theta_i(t)) &\sin(\theta_i(t))}^\top$ in \eqref{eq:state_z} is bounded for all time $t\geq 0$. Therefore, La Salle's invariance principle can be applied to analyze the asymptotic convergence of the autonomous system \eqref{eq:state_z_dot} 
to the largest invariant set $\Omega=\left \{ z:\dot{V}(z)=0 \right \}$. For this purpose, let us introduce a new variable to simplify the first term of \eqref{eq:V_dot_final} as follows
\begin{equation}\label{eq:E_i}
       E_i :=  \frac{\partial P_i}{\partial X_i}\nabla^2 J_{i} \vec{o}^{\top}_i
      = -\frac{\partial P_i}{\partial e_{ij}}\Vec{o}_{\mu_{ij}}\nabla^2 J_{i} \vec{o}^{\top}_i.
\end{equation}
From the Lyapunov inequality \eqref{eq:V_dot_final}, the state $z\in \Omega$ must satisfy 
\begin{equation}
\left.\begin{matrix}
    \begin{aligned}
       H_{ij} K_{ij} &= 0  \\
         E_i & = 0, 
    \end{aligned}
    \end{matrix}\right\}
\end{equation} 
for all $(i,j)\in\mathcal E$. The equations 
$H_{ij} K_{ij} = 0$ for all $(i,j)\in\mathcal{E}$ hold only when $e_{ij}=0$, 
since $P_i$ is defined as a strictly convex function with only one minimum at $e_{ij} = 0$ (i.e., $\frac{\partial P_i}{\partial e_{ij}} = 0 $ if and only if $e_{ij}=0$). 
With regards to the second condition $ E_i  = 0$, we first note $\vec{o}_{\mu_{ij}} \neq \bold{0}_{1\times 2}$ by definition. Hence $ E_i = 0$ in \eqref{eq:E_i} can be analyzed in the following two cases. 
\begin{itemize}
    \item Case 1: $\frac{\partial P_i}{\partial e_{ij}} = 0$ for all $(i,j)\in\mathcal{E}$. Given the definition of $P_i$, this case implies that 
    all connected pairs converge to the desired inter-agent spacing $D^*_{ij}$ and $e_{ij} = 0$. 
    \item Case 2: $\frac{\partial P_i}{\partial e_{ij}} \neq 0$ and $\left(\frac{\partial P_i}{\partial X_i}\nabla^2 J_{i}
    \right) \perp \vec{o}^{\top}_i $. Based on the fact  $\left \langle  \vec{o}_i , \vec{o}^\perp_i \right \rangle = 0$, the motion variables $\sbm{v_i \\ \omega_i}$ for agent $i$ in \eqref{eq:flocking_controller} can be calculated as
    \begin{equation}
    \begin{aligned}
        v_i &= -\frac{K_f}{N_i} \vec{o}_i \nabla^2 J_{i} \left(\frac{\partial P_i}{\partial X_i} \right)^\top  = 0\\
        \omega_{i} & = \frac{K_f}{N_i} \vec{o}^\perp_i  \nabla^2 J_{i} \left(\frac{\partial P_i}{\partial X_i}\right)^\top \neq 0.
    \end{aligned}
    \end{equation}
    In this case, the dynamics of agent $i$ in the $\Omega$-limit set satisfies $ \sbm{\dot x_i \\ \dot y_i \\ \omega_i} = \sbm{0 \\ 0 \\ \omega_i}$, which corresponds to the agent being stationary with 
    varying orientation $\vec{o}_i$ (or rotating). Given the fact that $P_i$ is a function of the spacing error $e_{ij}$, $\frac{\partial P_i}{\partial e_{ij}}$ remains constant so that 
    $\left(\frac{\partial P_i}{\partial X_i}\nabla^2 J_{i}
    \right) \perp \vec{o}^{\top}_i$ will not be maintained. Thus the positions associated to Case 2 do not belong to $\Omega$.
\end{itemize}

Therefore, the only solution of the invariant set $\Omega$ corresponds to the solution of $\frac{\partial P_i}{\partial e_{ij}} = 0$ for all $(i,j)\in\mathcal{E}$, which is given by zero 
inter-agent spacing error $e_{ij}=0$. 
This concludes the proof that the adaptive spacing policy-based flocking cohesion law \eqref{eq:flocking_controller} allows the adaptive scaling of the inter-agent’s gap to 
guarantee the well-posedness of geometrical flocking configuration, and the convergence of all agents to the well-posed flocking configuration. 
\end{proof}

\section{Application in Dynamic Interaction Topology}\label{sec:dynamic}
In the previous section, we consider the ASP-based flocking cohesion in the case that the undirected communication topology is static and connected for all time. In other words, there is no edge-connectedness loss during the evolution and the flocking cohesion can be achieved. However, the agent's limited sensing capability can cause switching communication topology in practical scenarios. The group splitting and fragmentation \cite{olfati2006flocking} can happen if the dynamic graph is not connected, and the flocking cohesion will fail.

To maintain the local connectivity of the switching graph and to achieve the flocking cohesion task, we present in this section the generalization of the ASP-based flocking cohesion scheme from the previous section to 
the dynamic communication topology. We consider also the presence of 
constraints, such as, connectivity preservation and inter-agent collision avoidance. 

Suppose that each agent $i$ has an identical maximum communication range $r$ and can only interact with the neighbors within the sensing zone. Let us define the neighbor set of agent $i$ at time $t$ as
\begin{equation}\label{eq:neighbor}
     \mathcal{N}_i(t) =  \left \{  j \neq i: c_{ij}(t) := r- \mu_{ij}(X_i(t),X_j(t))  > 0  \right \},
\end{equation}
where $\mu_{ij} $ is the gradient difference of the neighboring agents $i$ and $j$ in \eqref{eq:muij}. When the agents move around, the neighbor set $\mathcal N_i(t)$ can grow or contract, which reflects the fact that new edges to the agent $i$ are being established or removed due to the changes of proximity with respect to the other agents. Correspondingly, the interaction topology of the multi-agent unicycle system can be described by a dynamic undirected graph $\mathcal{G}(t)= (\mathcal{V},\mathcal{E}(t))$ whose edge set changes / switches 
at time sequence $\{t_{k}\}_{k\in\nline}$, where $\mathcal{E}(t) \subseteq  \left \{  (i,j): i\in \mathcal{V}, j \in \mathcal{N}_i(t) \right \}$ is the dynamic edge set. 
Note that, by the definition of $\mathcal N_i(t)$, when 
at a time point $t_k$, the value $c_{ij}(t)$ in \eqref{eq:neighbor} changes its sign from positive to negative, then the pair $(i,j)$ is no longer connected and 
the edge will be removed from the graph. Similarly, the dynamic graph allows for newly-added edges if an agent enters the sensing range of other agents during flocking. In the following definition, we will consider the connected preservation property of a dynamic undirected graph $\mathcal G(t)$ in this setting. 
\begin{definition}
{\bf (Connectivity Preservation)} The dynamic undirected graph $\mathcal{G}(t)$ achieves {\em connectivity preservation} if $\mathcal G(t)$ is connected for all time $t\geq 0$. 
\end{definition}

In the following, we will analyze the performance of our proposed ASP-based flocking cohesion in this dynamic graph setting. In particular, we show as well the graph connectivity preservation property of the closed-loop systems. Using the ASP design in \eqref{eq:sij}-\eqref{eq:d_dot}, we need to introduce an additional constraint on the scaling parameter $\lambda$ in \eqref{eq:sij} that depends on the maximum communication range $r$. 
For generalizing the ASP-based flocking cohesion to the dynamic interaction topology, 
we assume additional conditions on the potential function $P_i$ as considered in Subsection \ref{sec:Flocking_Spacing}, namely, 
$P_i$ is a nonnegative, differentiable function of $e_{ij}$ and $s_{ij}$ with $(i,j)\in\mathcal{E}$, such that
\begin{enumerate}
    \item $P_i\rightarrow 0$ as $e_{ij}:=\mu_{ij} - d_{\text{nom}}s_{ij} \rightarrow 0 $;
    \item $P_i \rightarrow \infty$ as $c_{ij}\rightarrow  0$ or $\mu_{ij}\rightarrow  0$,
\end{enumerate}
where $c_{ij}$ is as defined in the dynamic neighbor set \eqref{eq:neighbor}. An example for such potential function is 
\begin{equation}\label{eq:potential_constr}
        P_i = \frac{1}{2}\sum_{j\in\mathcal{N}_i} e^2_{ij}\left [\left(\text{ln} \left(c_{ij}\mu_{ij}\right) \right)^2  +1\right].
    \end{equation}
Note that the error dynamics and the adaptive spacing policy remain 
the same as in \eqref{eq:D*} and \eqref{eq:eij}. 
\begin{theorem}\label{thm:dynamic}
    Consider a multi-agent unicycle system \eqref{eq:unicycle_model}, which is initialized in a connected graph 
    $\mathcal{G}(t_0)$ and is driven by the adaptive spacing policy-based flocking cohesion algorithm \eqref{eq:flocking_controller} with the scaling multiplier $s_{ij}$ be as in \eqref{eq:sij} satisfying $0<\lambda \leq \ln{\frac{r}{d_{\text{nom}}}}$. 
    Then the multi-agent unicycle system achieves connectivity preservation and flexible flocking cohesion. 
\end{theorem}
\begin{proof}
    Following the same proof as in the Theorem \ref{thm:static}, 
    the extended-state dynamics and Lyapunov function can be given by \eqref{eq:state_z} and \eqref{eq:Lyapunov}. As the dynamic communication topology $\mathcal{G}(t)$ switches at the time sequences $\{t_{k}\}_{k\in\nline}$, we thus calculate the derivative of $V(z)$ with respect to the time $t \in\left [ t_k, t_{k+1}\right)$ as:
\begin{equation}\label{eq:V_dot_conn}
\begin{aligned}
     \dot V(z) & =  \sum_{i\in\mathcal{V}} \dot P_i + \underbrace{\sum_{i\in\mathcal{V}} \left( -\cos(\theta_i)\sin(\theta_i)\dot \theta  + \cos(\theta_i)\sin(\theta_i)\dot \theta \right)}_{=0}
     \\& = \sum_{(i,j)\in\mathcal{E}} \frac{\partial  P_i}{\partial  \mu_{ij}}\dot \mu_{ij} + \frac{\partial  P_i}{\partial s_{ij}}\dot s_{ij} 
\end{aligned}
\end{equation}
Similar to the proof of Theorem \ref{thm:static}, $P_i$ is defined as a differentiable function that reaches its minimum at $e_{ij} = \mu_{ij}-d_{\text{nom}} s_{ij}=0$ and the final expression of $\dot V(z)$ is the same form as \eqref{eq:V_dot_final}.
As $V(z)$ includes the sum of the potentials for each connected agent pair $(i,j)\in\mathcal{E}(t)$, the positive-definite and boundedness properties of $V(z)$ imply that $P_i$ is bounded for each agent $i$ at all time $t \in\left [ t_k, t_{k+1}\right)$. Therefore, $c_{ij}\rightarrow 0$ or $\mu_{ij}\rightarrow 0$ can not happen according to the second rule of the potential function construction (as $P_i \rightarrow \infty$ in this case), and it implies that the connected agents will never exceed the sensing range $r$ or collide with each other during this time interval. Therefore, the connected edge would be maintained at the switching time such that $(i,j)\in\mathcal{E}(t_k) \subseteq \mathcal{E}(t_{k+1})$.

Correspondingly, the same argument and boundedness property can be applied to any newly-added edge during flocking evolution where $(i,j)\notin \mathcal{E}(t_{k-1})$ and $(i,j)\in \mathcal{E}(t_{k})$. Once the new edge is connected, it would be maintained for the next time interval such that $(i,j)\in\mathcal{E}(t_{k}) \subseteq \mathcal{E}(t_{k+1})$ and the neighbor set $\mathcal{N}_i(t_{k})\subseteq \mathcal{N}_i(t_{k+1}) $. Therefore, if $\mathcal{G}(t_k)$ is connected, the switching communication graph $\mathcal{G}(t_{k+1}) = \left\{ \mathcal{V}, \mathcal{E}(t_{k+1})\right\}$ is guaranteed to be connected. The same argument can be applied to the initially connected graph $\mathcal{G}(t_0)$ for connectivity maintenance for all time $t\geq t_0$. 


In addition, as the agents remain in the maximum interaction range (i.e., $c_{ij}(t) = r -\mu_{ij}(t)>0, \forall t\geq t_0$), the desired inter-agent space should be set as $0<D^*_{ij}<r$. With the formulation of $s_{ij}$ and its variable $\dot d_{ij}$ in \eqref{eq:sij} and \eqref{eq:d_dot},
the desired space is constrained as  $D^*_{ij}=d_{\text{nom}}s_{ij} \in \left (d_{\text{nom}}e^{-\lambda} ,d_{\text{nom}}e^{\lambda}\right ) $ with a positive nominal value $d_{\text{nom}}$. To satisfy the constraints of connectivity preservation, the scaling multiplier $\lambda$ should be set up with the condition $d_{\text{nom}}e^{\lambda} \leq r$ which refers to $0<\lambda \leq \ln{\frac{r}{d_{\text{nom}}}}$. 

{Given the finite number of agents $\left | \mathcal{V} \right|$ and the connectivity maintenance property, the switching times of the dynamic graph $\mathcal{G}(t) = \left\{\mathcal{V}, \mathcal{E}(t)\right\}$ is correspondingly finite and will converge to a fixed graph $\mathcal{G}(t_p)$ at time instance $t_p>t_0$. 
For analyzing the flocking cohesion convergence, we can consider the time interval $t\in\left [ t_p, +\infty \right)$, where the communication graph $\mathcal{G}(t)=(\mathcal{V},\mathcal{E}(t_p))$ is fixed and time-invariant. In this case, the LaSalle invariant principle can be applied, and the convergence analysis follows the proof of the largest invariant set $\Omega=\left \{ z:\dot{V}(z)=0 \right \}$ in Theorem \ref{thm:static}.} This concludes the proof.
\end{proof}

Accordingly, with the same setup as in Theorem \ref{thm:dynamic}, we have the following technical lemmas.
\begin{lemma}\label{lem:collision}
  Given the potential function $P_i$ in \eqref{eq:potential_constr}, the inter-agent collision avoidance can be guaranteed during the flocking evolution.
\end{lemma}
\begin{proof}
    Note the inter-agent collision corresponds only to 
    pairs of agents moving in a close range. The flocking group can be considered as \emph{collision-free} if the connected agents are collision-free. As proved in Theorem \ref{thm:dynamic} that the potential function $P_i$ is bounded, we have $\mu_{ij}(t) \neq 0$ for any connected agents over the time $t\geq t_0$. Thus, inter-agent collision avoidance can be guaranteed for the whole group flocking (i.e. $\mu_{ij}(t)\neq 0$, $j\neq i \in\mathcal{V}$, $\forall t \geq t_0$).
\end{proof}


\section{Simulation Results}\label{sec:simulation}
For the simulation setup demonstrating the efficacy of our proposed algorithm, a group of five unicycle agents with an initially connected graph $\mathcal{G}$ is considered in {two different scalar fields $\mathbb{R}^2$ 
$J(x,y)= - x^2 - y^2$ and $J(x,y) = -y^3-2(x^2+y)$, respectively.} Each agent $i$ can only interact with its connected neighbor agents $j\in\mathcal{N}_i$ where they exchange their local field gradient measurements $X_i:=\nabla J_{i\in\mathcal{V}}$. For numerical simulation, we consider $d_{\text{nom}}=2$. 
In the case of a dynamic graph $\mathcal{G}(t)$, the agents have identical maximum sensing range $r=10$. 

In the first simulation result, a static undirected communication topology $\mathcal{G}$ is considered that corresponds to the multi-agent system setup in Theorem \ref{thm:static}. Figure \ref{static_nonq:init} shows a complete graph $\mathcal G$ and initial position of each agent, where the connected edges are shown in light-blue. 
It is clear that the five agents in a complete graph can not form a configuration where each pair of inter-agent is equally distanced, and therefore it requires flexible flocking cohesion control approach. Using the ASP-based flocking controller, all agents eventually converge to a rigid geometry as shown in Figure \ref{static_nonq:flex_final} where each agent's trajectories are denoted in distinct colors, and the corresponding spacing errors are plotted in Figure \ref{static_nonq:flex_error}. {As a comparison, Figure \ref{static_nonq:stand_final} demonstrates the standard flocking cohesion results with preset fixed inter-agent space $d_{\text{nom}}$. This final configuration is widely accepted as \textit{quasi $\alpha$-lattice} where each inter-agent space converges to a close range to $d_{\text{nom}}$ as shown in Figure \ref{static_nonq:stand_error}. } {According to the deviation energy function $E$ for describing the degree in which a configuration differs from a standard $\alpha$-\textit{Lattice} in \cite{olfati2006flocking}:
\begin{equation}
    E = \frac{1}{\left | \mathcal{E} \right |+1}\sum_{(i,j)\in\mathcal{E}}(\mu_{ij}-d_{\text{nom}})^2,
\end{equation}
the mean-square-errors of the final configuration to the desired $d_{\text{nom}}$ value in every edge are $E_f =0.57$ and $E_s=0.33$, for the ASP-based flexible flocking controller and standard flocking cohesion controller, respectively. The results are plotted in Figure \ref{static_nonq:e_mean} in blue and red lines. Both of them show that $E\leq \varepsilon ^2 d_{\text{nom}}^2$ with $\varepsilon \ll 1$, and it shows that the proposed flexible flocking can also achieve a low-deviation-energy configuration of standard $\alpha$-\textit{Lattice} while coordinating the agent pair to an adaptive space. In particular, flexible flocking control approach results in a faster convergence speed than 
the standard one as seen in Figure \ref{static_nonq:e_mean}. In order to demonstrate the convergence rate, a mean square error with respect to the ASP-based inter-agent space $D^*_{ij}$ in the flexible flocking method is additionally plotted in orange, where we define 
\begin{equation}
    E_{ASP} = \frac{1}{\left | \mathcal{E} \right |+1}\sum_{(i,j)\in\mathcal{E}}(\mu_{ij}-D^*_{ij})^2.
\end{equation}
In other words, the strategy of introducing adaptive inter-agent spacing policy to the flocking cohesion not only allows the multi-agent systems to fastly reach an admissible configuration that is close to the desired one ($\alpha$-\textit{Lattice}), but also provides flexibility for the agent motion and geometry organization. Note this property can be crucial for various tasks (e.g., obstacle avoidance, large-scale coordination, adaptive cruise control (ACC) systems, lane change, etc.).} 

  \begin{figure*}[htbp]
	\centering
         \subfigure[]{
		\begin{minipage}[t]{0.3\textwidth}
		    \label{static_nonq:init}
			\centering			
                \includegraphics[width=1\textwidth]{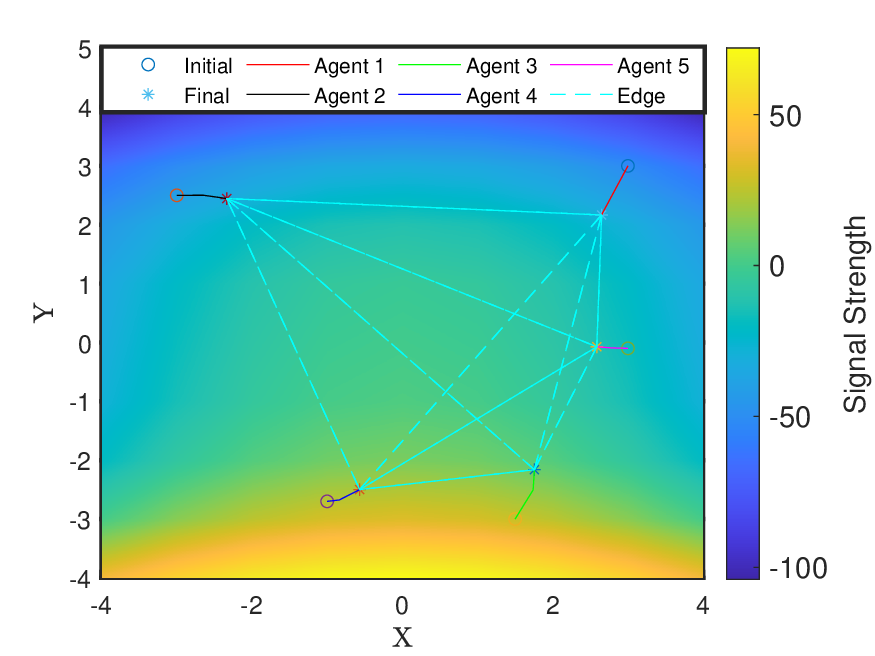}
                \end{minipage}%
	}%
	\subfigure[]{
		\begin{minipage}[t]{0.3\textwidth}
		 \label{static_nonq:flex_final}
			\centering			
                \includegraphics[width=1\textwidth]{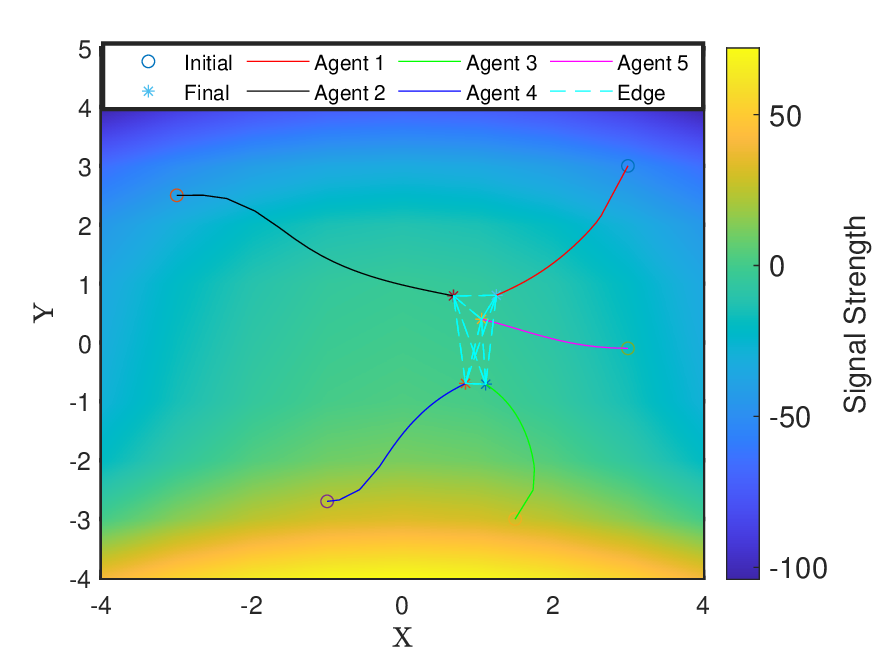}
		\end{minipage}%
	}%
        \subfigure[]{
		\begin{minipage}[t]{0.3\textwidth}
		 \label{static_nonq:flex_error}
			\centering			
                \includegraphics[width=1\textwidth]{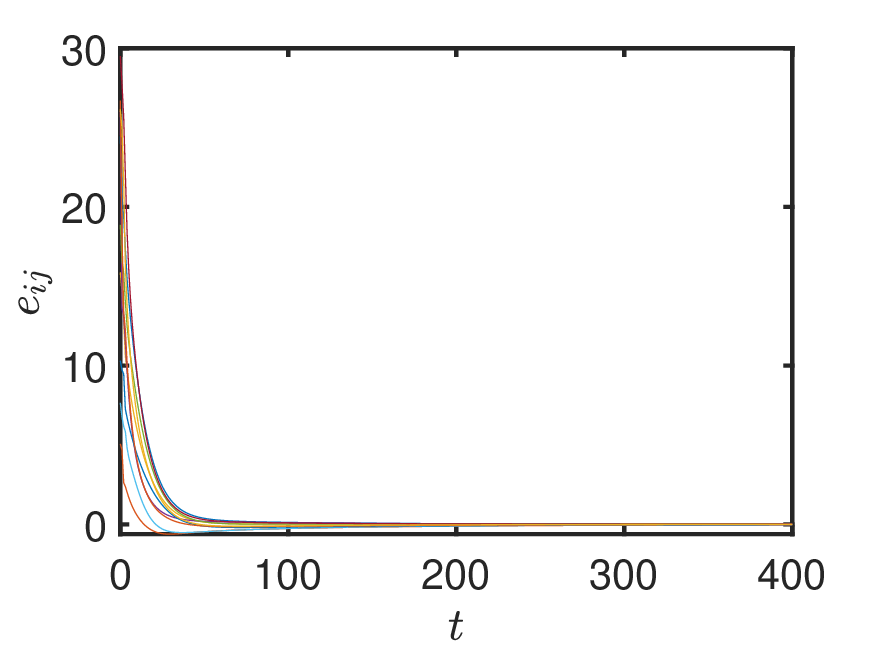}
		\end{minipage}%
	}%

  \subfigure[]{
		\begin{minipage}[t]{0.3\textwidth}
		 \label{static_nonq:stand_final}
			\centering			
                \includegraphics[width=1\textwidth]{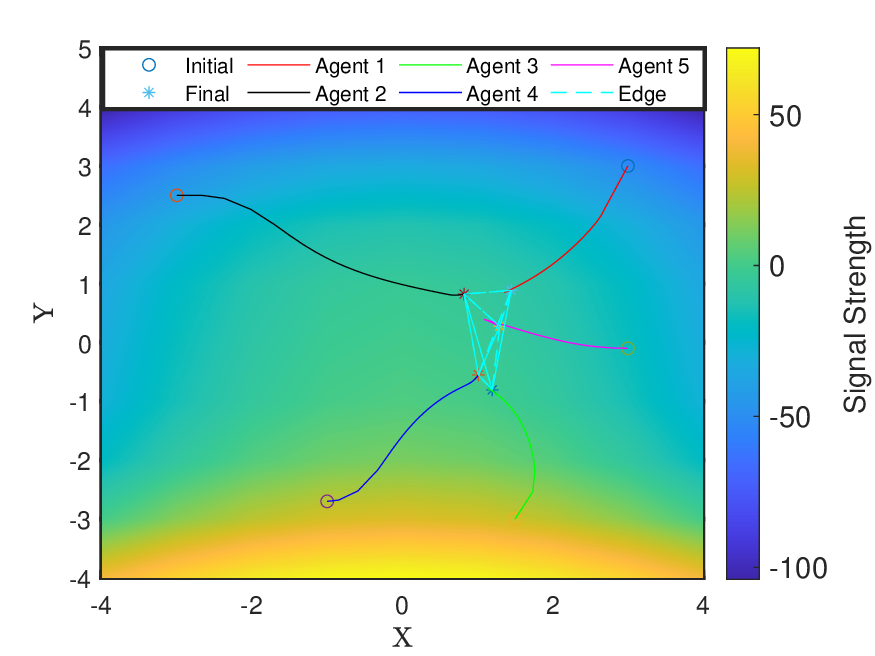}
		\end{minipage}%
	}%
     \subfigure[]{
		\begin{minipage}[t]{0.3\textwidth}
		 \label{static_nonq:stand_error}
			\centering			
                \includegraphics[width=1\textwidth]{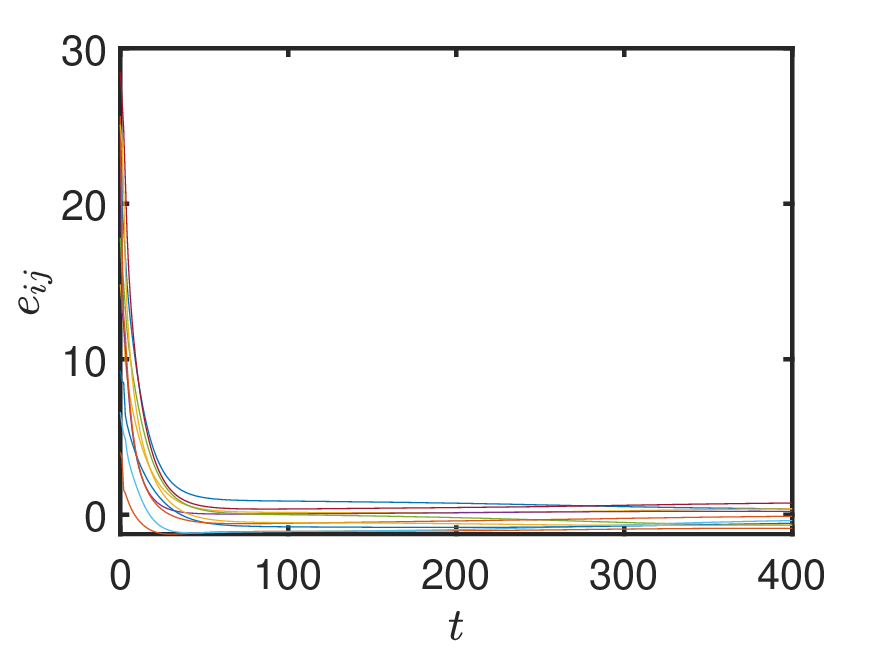}
		\end{minipage}%
	}%
  \subfigure[]{
		\begin{minipage}[t]{0.3\textwidth}
		 \label{static_nonq:e_mean}
			\centering			
                \includegraphics[width=1\textwidth]{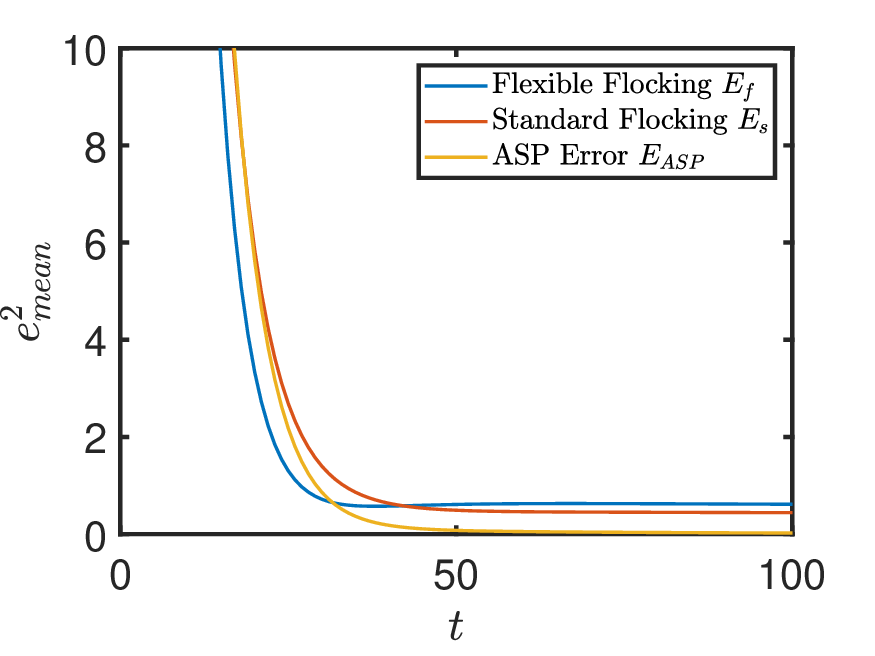}
		\end{minipage}%
	}%

	\centering
 \caption{Comparision of the flocking results of multi-unicycle system ($N=5$) with static communication graph $\mathcal{G}$ in the field $J(x,y) = -y^3-2(x^2+y)$. (a) Initial complete graph. (b)-(c) ASP-based flexible flocking configuration and spacing error $e_{ij} = \mu_{ij}-D^*_{ij}$. The parameters are set as $\lambda=1$ and $d_{ij}(t_0) = 0$.  (d)-(e) Standard flocking with fixed inter-agent space $d_{\text{nom}}=2$ and the spacing error $e_{ij} = \mu_{ij}-d_{\text{nom}}$. {(f) Mean-square-spacing-error comparison. Note $E_f$ and $E_s$ are with $ E = \frac{1}{\left | \mathcal{E} \right |+1}\sum_{(i,j)\in\mathcal{E}}(\mu_{ij}-d_{\text{nom}})^2$, and $E_{ASP} =\frac{1}{\left | \mathcal{E} \right |+1}\sum_{(i,j)\in\mathcal{E}}(\mu_{ij}-D^*_{ij})^2.$} }
	\label{fig:static_nonq}
\end{figure*}

  \begin{figure*}[htbp]
	\centering
         \subfigure[]{
		\begin{minipage}[t]{0.3\textwidth}
		    \label{dynamic_non:init}
			\centering			
                \includegraphics[width=1\textwidth]{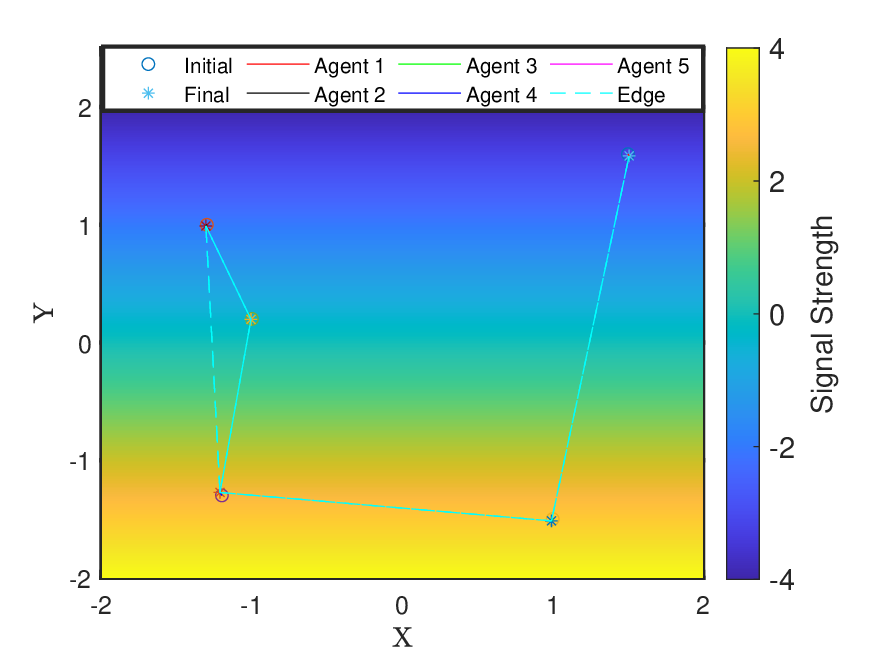}
                \end{minipage}%
	} 
	\subfigure[]{
		\begin{minipage}[t]{0.3\textwidth}
		 \label{dynamic_non:flex_final}
			\centering			
                \includegraphics[width=1\textwidth]{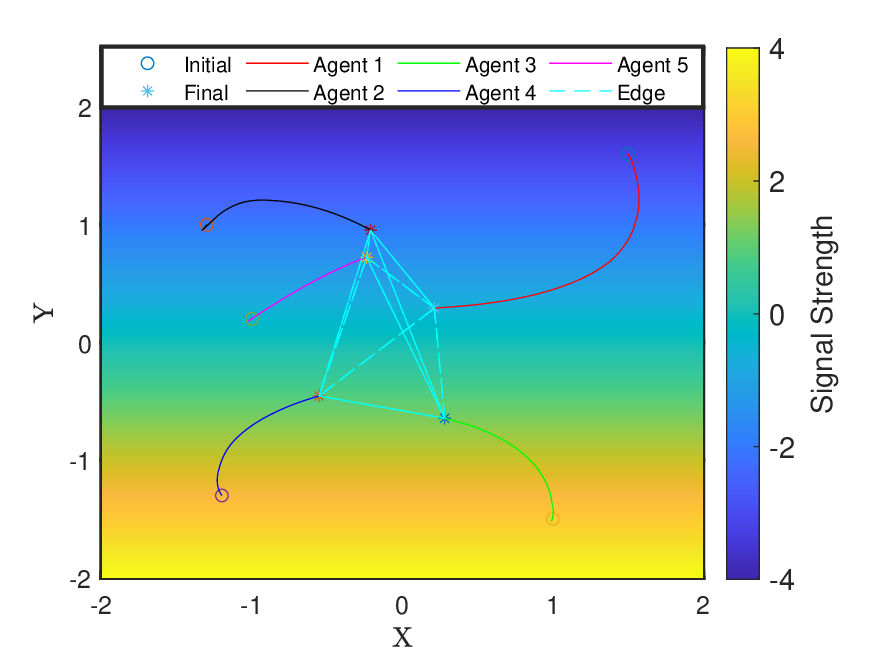}
		\end{minipage}%
	}%
        \subfigure[]{
		\begin{minipage}[t]{0.3\textwidth}
		 \label{dynamic_non:flex_error}
			\centering			
                \includegraphics[width=1\textwidth]{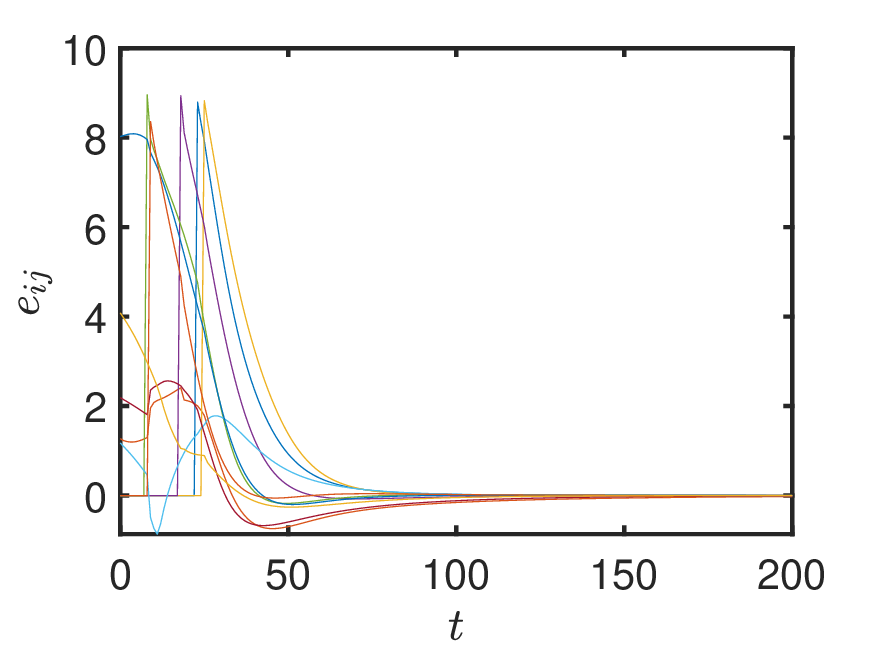}
		\end{minipage}%
	}%

     \subfigure[]{
		\begin{minipage}[t]{0.3\textwidth}
		 \label{dynamic_non:flex_muij}
			\centering			
                \includegraphics[width=1\textwidth]{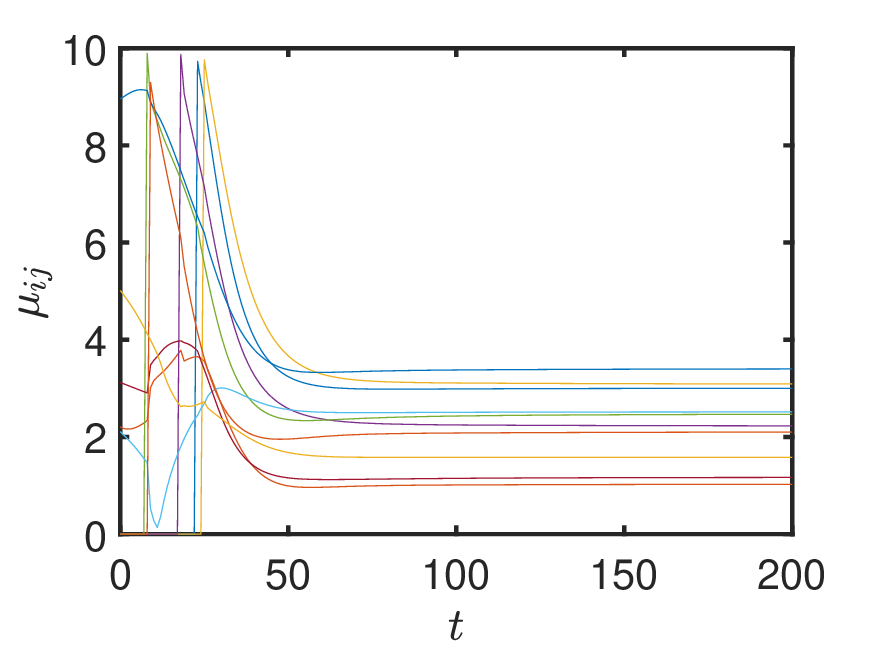}
		\end{minipage}%
	}%
     \subfigure[]{
		\begin{minipage}[t]{0.3\textwidth}
		 \label{dynamic_non:flex_sij}
			\centering			
                \includegraphics[width=1\textwidth]{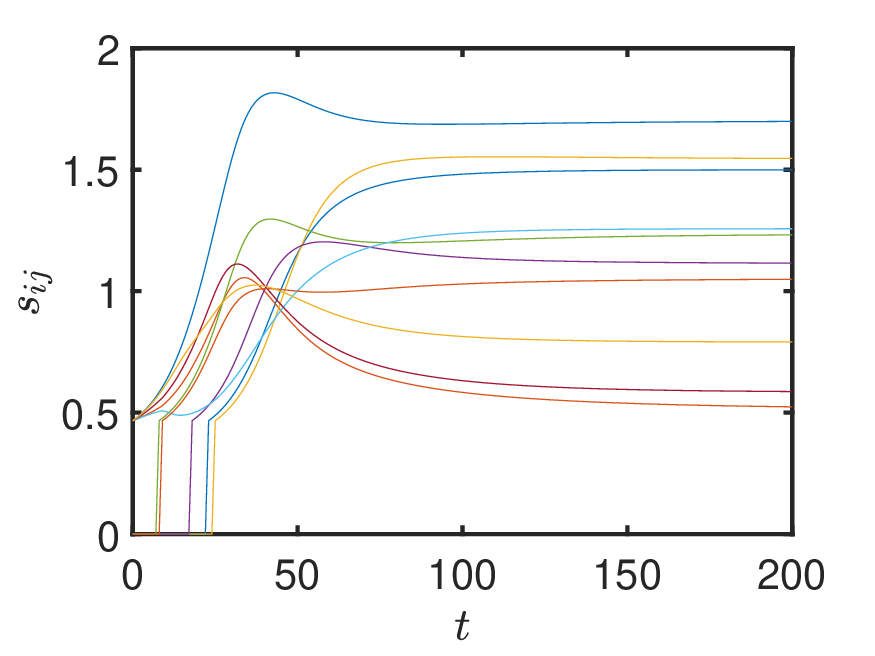}
		\end{minipage}%
	}%

	\centering
     \caption{Flexible flocking results of multi-agent system with dynamic communication graph $\mathcal{G}(t)$ where the maximum interaction range is set as $r=10$ for all agents. (a) Initial incomplete graph. (b) Final flocking configuration. (c) Inter-agent spacing error $e_{ij} = \mu_{ij}-D^*_{ij}$. (d) Inter-agent space $\mu_{ij}$. (e) Scaling multiplier $s_{ij}$.}
	\label{fig:dynamic_non_flex}
\end{figure*}

Let us now evaluate the proposed ASP flocking cohesion algorithm as in Theorem \ref{thm:dynamic} 
in a dynamic graph $\mathcal{G}(t)$. In this case, we use a connected and incomplete graph for the initial graph $\mathcal G(0)$ as shown in Figure \ref{dynamic_non:init}. Particularly, the nodes within the maximum interaction range $r$ are considered connected and the corresponding edges are shown in light-blue. During the flocking evolution, the spacing error $e_{ij}$ is only drawn for the connected edge $(i,j)\in\mathcal{E}(t)$ at each switching time
(see the initial flat states in Figure \ref{dynamic_non:flex_error}). It can be seen that the initially unconnected agents started entering the interaction area around time $t\approx 7$ (as shown by the jumps in Figure \ref{dynamic_non:flex_error} where the corresponding scaling parameter $s_{ij}$ becomes active in Figure \ref{dynamic_non:flex_sij}). The trajectories of agents are shown in Figure \ref{dynamic_non:flex_final}, where the final interaction topology is fixed and complete, rendering the flocking group convergences to a rigid geometry (see the spacing error convergence in Figure \ref{dynamic_non:flex_error}). The corresponding desired spacing value $D^*_{ij}$ differs from each pair of connected edges by various $s^*_{ij}$ in Figure \ref{dynamic_non:flex_sij}. Figure \ref{dynamic_non:flex_muij} shows that there is no inter-agent collision and no breaking of connectivity as $ 0<\mu_{ij}(x,y)<r$ for all $i\neq j \in\mathcal{V}$ during the motion. This verifies the efficacy of Theorem \ref{thm:dynamic} and Lemma \ref{lem:collision}. 



\section{Conclusion}\label{sec:conclusion}
This paper presents an adaptive spacing policy for 
solving distributed flexible flocking cohesion control problem in 
multi-agent systems. It is shown that the policy provides flexibility for the collective motion of large-scale agents in reaching a consensus on admissible geometrical shape and can be easily modified to incorporate complex constraints (e.g., inter-agent collision). The fragmentation phenomenon in leaderless flocking tasks can be avoided due to the connectivity preservation property of the proposed cohesion law. In this distributed control method, only limited local information (field gradient) exchange is required between neighboring agents, which relaxes the inter-agent real-time communication requirement for a large-scale deployment of multi-agent systems. 

\bibliographystyle{IEEEtran}
\bibliography{bibtex/bib/IEEEabrv.bib, bibtex/bib/reference.bib}{}

\begin{thebibliography}{10}
\providecommand{\url}[1]{#1}
\csname url@samestyle\endcsname
\providecommand{\newblock}{\relax}
\providecommand{\bibinfo}[2]{#2}
\providecommand{\BIBentrySTDinterwordspacing}{\spaceskip=0pt\relax}
\providecommand{\BIBentryALTinterwordstretchfactor}{4}
\providecommand{\BIBentryALTinterwordspacing}{\spaceskip=\fontdimen2\font plus
\BIBentryALTinterwordstretchfactor\fontdimen3\font minus
  \fontdimen4\font\relax}
\providecommand{\BIBforeignlanguage}[2]{{%
\expandafter\ifx\csname l@#1\endcsname\relax
\typeout{** WARNING: IEEEtran.bst: No hyphenation pattern has been}%
\typeout{** loaded for the language `#1'. Using the pattern for}%
\typeout{** the default language instead.}%
\else
\language=\csname l@#1\endcsname
\fi
#2}}
\providecommand{\BIBdecl}{\relax}
\BIBdecl

\bibitem{Moshtagh-i}
N.~Moshtagh and A.~Jadbabaie, ``Distributed geodesic control laws for flocking
  of nonholonomic agents,'' \emph{IEEE Transactions on Automatic Control},
  vol.~52, no.~4, pp. 681--686, 2007.

\bibitem{Saulnier}
K.~Saulnier, D.~Saldana, A.~Prorok, G.~J. Pappas, and V.~Kumar, ``Resilient
  flocking for mobile robot teams,'' \emph{IEEE Robotics and Automation
  Letters}, vol.~2, no.~2, pp. 1039--1046, 2017.

\bibitem{besselink}
B.~Besselink and K.~H. Johansson, ``Control of platoons of heavy-duty vehicles
  using a delay-based spacing policy,'' \emph{IFAC-PapersOnLine}, vol.~48,
  no.~12, pp. 364--369, 2015.

\bibitem{turri}
V.~Turri, B.~Besselink, and K.~H. Johansson, ``Cooperative look-ahead control
  for fuel-efficient and safe heavy-duty vehicle platooning,'' \emph{IEEE
  Transactions on Control Systems Technology}, vol.~25, no.~1, pp. 12--28,
  2016.

\bibitem{sun2021collaborative}
Z.~Sun, H.~Garcia~de Marina, B.~D. Anderson, and C.~Yu, ``Collaborative
  target-tracking control using multiple fixed-wing unmanned aerial vehicles
  with constant speeds,'' \emph{Journal of Guidance, Control, and Dynamics},
  vol.~44, no.~2, pp. 238--250, 2021.

\bibitem{de2019flexible}
H.~G. De~Marina and E.~Smeur, ``Flexible collaborative transportation by a team
  of rotorcraft,'' in \emph{2019 International Conference on Robotics and
  Automation (ICRA)}.\hskip 1em plus 0.5em minus 0.4em\relax IEEE, 2019, pp.
  1074--1080.

\bibitem{Frasca}
M.~Frasca, A.~Buscarino, A.~Rizzo, L.~Fortuna, and S.~Boccaletti,
  ``Synchronization of moving chaotic agents,'' \emph{Physical Review Letters},
  vol. 100, no.~4, p. 044102, 2008.

\bibitem{Reynolds}
C.~W. Reynolds, ``Flocks, herds and schools: A distributed behavioral model,''
  in \emph{Proceedings of the 14th annual Conference on Computer Graphics and
  Interactive Techniques}, 1987, pp. 25--34.

\bibitem{Kearns}
D.~B. Kearns, ``A field guide to bacterial swarming motility,'' \emph{Nature
  Reviews Microbiology}, vol.~8, no.~9, pp. 634--644, 2010.

\bibitem{Okubo}
A.~Okubo, ``Dynamical aspects of animal grouping: swarms, schools, flocks, and
  herds,'' \emph{Advances in Biophysics}, vol.~22, pp. 1--94, 1986.

\bibitem{Zavlanos}
M.~M. Zavlanos, A.~Jadbabaie, and G.~J. Pappas, ``Flocking while preserving
  network connectivity,'' in \emph{2007 46th IEEE Conference on Decision and
  Control (CDC)}.\hskip 1em plus 0.5em minus 0.4em\relax IEEE, 2007, pp.
  2919--2924.

\bibitem{Loizou}
S.~G. Loizou, D.~G. Lui, A.~Petrillo, and S.~Santini, ``Connectivity preserving
  formation stabilization in an obstacle-cluttered environment in the presence
  of time-varying communication delays,'' \emph{IEEE Transactions on Automatic
  Control}, vol.~67, no.~10, pp. 5525--5532, 2021.

\bibitem{olfati2006flocking}
R.~Olfati-Saber, ``Flocking for multi-agent dynamic systems: Algorithms and
  theory,'' \emph{IEEE Transactions on Automatic Control}, vol.~51, no.~3, pp.
  401--420, 2006.

\bibitem{Deghat}
M.~Deghat, B.~D. Anderson, and Z.~Lin, ``Combined flocking and distance-based
  shape control of multi-agent formations,'' \emph{IEEE Transactions on
  Automatic Control}, vol.~61, no.~7, pp. 1824--1837, 2015.

\bibitem{ploeg2011}
J.~Ploeg, B.~T. Scheepers, E.~Van~Nunen, N.~Van~de Wouw, and H.~Nijmeijer,
  ``Design and experimental evaluation of cooperative adaptive cruise
  control,'' in \emph{2011 14th International IEEE Conference on Intelligent
  Transportation Systems (ITSC)}.\hskip 1em plus 0.5em minus 0.4em\relax IEEE,
  2011, pp. 260--265.

\bibitem{ploeg2013}
J.~Ploeg, N.~Van De~Wouw, and H.~Nijmeijer, ``Lp string stability of cascaded
  systems: Application to vehicle platooning,'' \emph{IEEE Transactions on
  Control Systems Technology}, vol.~22, no.~2, pp. 786--793, 2013.

\bibitem{Ames_CDC}
A.~D. Ames, J.~W. Grizzle, and P.~Tabuada, ``Control barrier function based
  quadratic programs with application to adaptive cruise control,'' in
  \emph{53rd IEEE Conference on Decision and Control}.\hskip 1em plus 0.5em
  minus 0.4em\relax IEEE, 2014, pp. 6271--6278.

\bibitem{Wang_CBF}
L.~Wang, A.~D. Ames, and M.~Egerstedt, ``Multi-objective compositions for
  collision-free connectivity maintenance in teams of mobile robots,'' in
  \emph{2016 IEEE 55th Conference on Decision and Control (CDC)}.\hskip 1em
  plus 0.5em minus 0.4em\relax IEEE, 2016, pp. 2659--2664.

\bibitem{Capelli_2020}
B.~Capelli and L.~Sabattini, ``Connectivity maintenance: Global and optimized
  approach through control barrier functions,'' in \emph{2020 IEEE
  International Conference on Robotics and Automation (ICRA)}.\hskip 1em plus
  0.5em minus 0.4em\relax IEEE, 2020, pp. 5590--5596.

\bibitem{Capelli_2021}
B.~Capelli, H.~Fouad, G.~Beltrame, and L.~Sabattini, ``Decentralized
  connectivity maintenance with time delays using control barrier functions,''
  in \emph{2021 IEEE International Conference on Robotics and Automation
  (ICRA)}.\hskip 1em plus 0.5em minus 0.4em\relax IEEE, 2021, pp. 1586--1592.

\bibitem{Li_Flocking}
T.~Li and B.~Jayawardhana, ``Collision-free source seeking and flocking control
  of multi-agents with connectivity preservation,'' \emph{arXiv preprint
  arXiv:2301.04576}, 2023.

\bibitem{Ames_TAC}
A.~D. Ames, X.~Xu, J.~W. Grizzle, and P.~Tabuada, ``Control barrier function
  based quadratic programs for safety critical systems,'' \emph{IEEE
  Transactions on Automatic Control}, vol.~62, no.~8, pp. 3861--3876, 2016.

\bibitem{Xiao_auto}
W.~Xiao, C.~A. Belta, and C.~G. Cassandras, ``Sufficient conditions for
  feasibility of optimal control problems using control barrier functions,''
  \emph{Automatica}, vol. 135, p. 109960, 2022.

\bibitem{Jia}
Y.~Jia and L.~Wang, ``Decentralized formation flocking for multiple
  non-holonomic agents,'' in \emph{2013 IEEE Conference on Cybernetics and
  Intelligent Systems (CIS)}.\hskip 1em plus 0.5em minus 0.4em\relax IEEE,
  2013, pp. 100--105.

\bibitem{Wen}
G.~Wen, Z.~Duan, H.~Su, G.~Chen, and W.~Yu, ``A connectivity-preserving
  flocking algorithm for multi-agent dynamical systems with bounded potential
  function,'' \emph{IET Control Theory \& Applications}, vol.~6, no.~6, pp.
  813--821, 2012.

\bibitem{Li_bounded}
X.~Li, D.~Sun, and J.~Yang, ``A bounded controller for multirobot navigation
  while maintaining network connectivity in the presence of obstacles,''
  \emph{Automatica}, vol.~49, no.~1, pp. 285--292, 2013.

\bibitem{Ajorlou}
A.~Ajorlou, A.~Momeni, and A.~G. Aghdam, ``A class of bounded distributed
  control strategies for connectivity preservation in multi-agent systems,''
  \emph{IEEE Transactions on Automatic Control}, vol.~55, no.~12, pp.
  2828--2833, 2010.

\bibitem{Kim}
Y.~Kim and M.~Mesbahi, ``On maximizing the second smallest eigenvalue of a
  state-dependent graph laplacian,'' in \emph{Proceedings of the 2005, American
  Control Conference, 2005.}\hskip 1em plus 0.5em minus 0.4em\relax IEEE, 2005,
  pp. 99--103.

\bibitem{Poonawala}
H.~A. Poonawala, A.~C. Satici, H.~Eckert, and M.~W. Spong, ``Collision-free
  formation control with decentralized connectivity preservation for
  nonholonomic-wheeled mobile robots,'' \emph{IEEE Transactions on Control of
  Network Systems}, vol.~2, no.~2, pp. 122--130, 2014.

\bibitem{sastry2013nonlinear}
S.~Sastry, \emph{Nonlinear systems: analysis, stability, and control}.\hskip
  1em plus 0.5em minus 0.4em\relax Springer Science \& Business Media, 2013,
  vol.~10.

\end{thebibliography}

\end{document}